\documentclass{article}
\usepackage{fullpage,amsmath,amsthm,amssymb,hyperref,cleveref,graphicx,url,rotating}
\usepackage{cancel}
\usepackage{stackrel}
\usepackage{mathtools} 
\usepackage{extarrows}

\graphicspath{{Fig/}{./}}

\newtheorem{Proposition}{Proposition}
\newtheorem{Example}{Example}
\newtheorem{Remark}{Remark}
\newtheorem{Definition}{Definition}

\def\skl{\mathrm{skl}}
\def\calF{\mathcal{F}}
\def\Sym{\mathrm{Sym}}
\def\calB{\mathcal{B}}
\def\calD{\mathcal{D}}
\def\det{\mathrm{det}}
\def\tr{\mathrm{tr}}
\def\id{\mathrm{id}}
\def\calE{\mathcal{E}}

\def\calX{\mathcal{X}}

\def\bbN{\mathbb{N}}

\def\dmu{\mathrm{d}\mu}
\def\dP{\mathrm{d}P}

\def\Inner#1#2{{\left\langle#1,#2\right\rangle}}
\def\inner#1#2{{\left\langle#1,#2\right\rangle}}

\def\calM{\mathcal{M}}
\def\calF{\mathcal{F}}

\def\calE{\mathcal{E}}

\def\bartheta{{\bar\theta}}

\def\dmu{\mathrm{d}\mu}

\def\tp{{\tilde{p}}}
\def\tq{{\tilde{q}}}
\def\KL{\mathrm{KL}}
\def\dlambda{\mathrm{d}\lambda}
\def\calA{\mathcal{A}}

\def\bbR{\mathbb{R}}

\def\st{\ :\ }

\def\calP{\mathcal{P}}

\title{Divergences induced by dual subtractive and divisive normalizations of exponential families and their convex deformations}

 \date{}

\author{Frank Nielsen\\ \ \\ Sony Computer Science Laboratories Inc, Japan}

\begin{document}

\maketitle

\begin{abstract}
Exponential families are statistical models which are the workhorses in statistics, information theory, and machine learning among others.
An exponential family can either be normalized subtractively by its cumulant or free energy function or equivalently normalized divisively by its partition function.
Both subtractive and divisive normalizers are strictly convex and smooth functions inducing pairs of Bregman and Jensen divergences.
It is well-known that skewed Bhattacharryya distances between probability densities of an exponential family amounts to
skewed Jensen divergences induced by the cumulant function between their corresponding natural parameters, and
 in limit cases that the sided Kullback-Leibler divergences amount to reverse-sided Bregman divergences.
In this paper, we first show that the $\alpha$-divergences between unnormalized densities of an exponential family amounts to 
scaled $\alpha$-skewed Jensen divergences induced by the partition function.
We then show how comparative convexity with respect to a pair of quasi-arithmetic means allows to deform both convex functions and their arguments, and thereby define dually flat spaces with corresponding divergences when ordinary convexity is preserved.
\end{abstract}


\section{Introduction}\label{sec:intro}

In information geometry~\cite{IG-2016}, any strictly convex and smooth function induces a dually flat space (DFS) with a canonical divergence  which can be expressed in charts either as dual Bregman divergences~\cite{Bregman-1967} or equivalently as dual Fenchel-Young divergences~\cite{nielsen2019monte}.
For example, the cumulant function of an exponential family~\cite{brown1986fundamentals} (also called free energy) generates a DFS: An exponential family manifold~\cite{scarfone2014legendre} with the canonical divergence yielding the reverse Kullback-Leibler divergence.
Another typical example is the negative entropy (negentropy for short) of a mixture family inducing a DFS: A mixture family manifold with the canonical divergence yielding the Kullback-Leibler divergence~\cite{nielsen2019monte}.
Any  strictly convex and smooth function also induces a family of scaled skewed Jensen divergences~\cite{zhang2004divergence,nielsen2011burbea} which includes in limit cases the sided forward and reverse Bregman divergences.

In \S\ref{sec:EF}, we show that there are two equivalent approaches to normalize an exponential family: Either by its cumulant function or by its partition function. Since both cumulant and partition functions are convex, they induce (i) families of  scaled skewed Jensen divergences and (ii) dually flat spaces with corresponding Bregman divergences corresponding to statistical divergences (Proposition~\ref{prop:pair}).
In \S\ref{sec:divF}, we recall the well-known result that the statistical $\alpha$-skewed Bhattacharyya distances between {\em probability  densities} of  an exponential family  amount to a scaled $\alpha$-skewed Jensen divergence between their natural parameters.
In \S\ref{sec:divZ}, we prove that the $\alpha$-divergences between   {\em unnormalized densities} of a exponential family amount to scaled $\alpha$-skewed Jensen divergence between their natural parameters (Proposition~\ref{prop:Z}).
More generally, we explain  in section~\ref{sec:deforming}  how to deform a convex function using comparative convexity~\cite{niculescu2006convex}:
 When ordinary convexity of the deformed convex function is preserved, we thus get new  skew Jensen divergences and DFSs with corresponding Bregman divergences.
Section~\ref{sec:concl} concludes this work with a discussion.

\section{Dual subtractive and divisive normalizations of exponential families}\label{sec:EF}

Let $(\calX,\calA,\mu)$ be a measure space where $\calX$ denotes the sample set (e.g., finite alphabet, $\bbN$, $\bbR^d$, space of positive-definite matrices $\Sym_{++}(d)$, etc.), $\calA$ a $\sigma$-algebra on $\calX$ (e.g., power set $2^\calX$, Borel $\sigma$-algebra $\calB(\calX)$, etc.), and $\mu$ a positive measure (e.g., counting measure or Lebesgue measure) on the measurable space $(\calX,\calA)$.

An exponential family~\cite{brown1986fundamentals} is a set of probability distributions $\calP=\{P_\lambda \st \lambda\in\Lambda\}$ all dominated by $\mu$ such that their Radon-Nikodym densities $p_\lambda(x)=\frac{\dP_\lambda}{\dmu}(x)$  can be expressed as
$$
p_\lambda(x) \propto \tp_\lambda(x) =\exp\left(\Inner{\theta(\lambda)}{t(x)}\right)  \, h(x),
$$
where $\theta(\lambda)$ is the natural parameter, $t(x)$ is the sufficient statistic, 
$h(x)$ is an auxiliary term used to define the base measure with respect to $\mu$, and
$\inner{\cdot}{\cdot}$ an inner product defined on the parameter space $\Lambda$.
The unnormalized positive density $\tp_\lambda(x)$ is indicated with a tilde notation, and the normalized probability density is obtained as
$p_\lambda(x)= \frac{1}{Z(\lambda)}\tp_\lambda(x)$ where $Z(\lambda)=\int \tp_\lambda(x)\dmu(x)$ is the Laplace transform of $\mu$. 

Exponential families include many well-known distributions: For example, the Bernoulli distributions, the Gaussian or normal distributions, the Gamma and Beta distributions, the Wishart distributions, the Poisson distributions, the Rayleigh distributions, etc. 
The categorical distributions (i.e., discrete distributions on a finite alphabet sample space) form an exponential family too, and exponential families can universally model arbitrarily closely any smooth density~\cite{CobbPEF-1983,nielsen2021fast}.
Furthermore, any two probability measures $Q$ and $R$ with densities $q$ and $r$ with respect to dominating measure $\mu=\frac{Q+R}{2}$ define an exponential family :
$$
\calP_{Q,R}=\{ p_\lambda(x) \propto q^\lambda(x)r^{1-\lambda}(x) \st \lambda\in (0,1)\},
$$ 
called the likelihood ratio exponential family~\cite{grunwald2007minimum} because the sufficient statistic is $t(x)=\log\frac{q(x)}{r(x)}$ (with auxiliary carrier term $h(x)=r(x)$) or the Bhattacharryya arc since the cumulant function of $\calP_{Q,R}$ is expressed as the negative of the  skewed Bhattacharryya distances~\cite{kailath1967divergence,nielsen2011burbea}.
In machine learning, undirected graphical models~\cite{wainwright2008graphical} and energy-based models~\cite{lecun2006tutorial} including Markov random fields~\cite{kindermann1980markov} and conditional random fields are
 exponential families~\cite{dai2019exponential}.

Let $\tp_\theta(x)=\exp\left(\inner{\theta}{t(x)}\right)  \, h(x)$ denote the unnormalized density expressed using the natural parameter $\theta=\theta(\lambda)$.
We can normalize $\tp_\theta(x)$ either by using the partition function $Z(\theta)$ or equivalently by using the cumulant function $F(\theta)$ as follows:
\begin{eqnarray}
p_\theta(x) &=& \frac{\exp(\inner{\theta}{t(x)})}{Z(\theta)}  \, h(x),\label{eq:divnorm} \\ 
&=& \exp\left(\inner{\theta}{t(x)}-F(\theta)+k(x)\right), \label{eq:subnorm}
\end{eqnarray}
where $h(x)=\exp(k(x))$, $Z(\theta)=\int \tp_\theta(x) \dmu(x)$ and $F(\theta)=\log Z(\theta)=\log \int \tp_\theta(x) \dmu(x)$.
Thus the logarithm and exponential functions allows to convert to and from the dual normalizers $Z$ and $F$:
$$
Z(\theta)=\exp(F(\theta))  \Leftrightarrow F(\theta)=\log Z(\theta).
$$
We may view Eq.~\ref{eq:divnorm} as an exponential tilting~\cite{efron2022exponential} of density $h(x)\dmu(x)$.

In the context of $\lambda$-deformed exponential families~\cite{zhang2021lambda} which generalize exponential families, 
function $Z(\theta)$ is called the divisive normalization factor (Eq.~\ref{eq:divnorm}) and function $F(\theta)$ is called the subtractive normalization factor (Eq.~\ref{eq:subnorm}). 
Notice that function $F(\theta)$ is called the cumulant function because when $X\sim p_\theta(x)$ is a random variable following a probability distribution of an exponential family,
 function $F(\theta)$ appears in the cumulant generating function of $X$: $K_X(t)=\log E_X[e^{\inner{t}{X}}]=F(\theta+t)-F(\theta)$.
The cumulant function is also called the log-normalizer or log-partition function in statistical physics.
Since $Z>0$ and $F=\log Z$, we deduce that $F\geq Z$  since $\log x\leq x$ for $x>0$.

The order $m$ of an exponential family is the dimension of the sufficient statistic $t(x)$.
The full natural parameter space $\Theta\in\bbR^m$ is defined by $\Theta=\{\theta \st  \int \tp_\theta(x) \dmu(x) <\infty\}$.
The exponential family is said full when $\theta$ is allowed to range in the full parameter space $\Theta$ (curved exponential families~\cite{IG-2016} have natural parameters $c(\theta)$ ranging in a subset of $\Theta$), and regular when $\Theta$ is an open (convex) domain.
A full regular exponential family such that both $t(x)=x$ and $k(x)=0$ is called a natural exponential family.
Without loss of generality, we consider natural exponential families or exponential families with $k(x)=0$ in the remainder.

It is well-known that the cumulant function $F(\theta)$ is a strictly convex function and that the partition function $Z(\theta)$ is strictly log-convex~\cite{barndorff2014information} (see Appendix~\ref{sec:appendix}):

\begin{Proposition}[\cite{barndorff2014information}]\label{prop:naturalconvex}
The natural parameter space $\Theta$ of an exponential family is convex.
\end{Proposition}

\begin{Proposition}[\cite{barndorff2014information}]\label{prop:FZconvexity} 
The  cumulant function $F(\theta)$ is strictly convex and the partition function $Z(\theta)$  is positive and strictly log-convex.
\end{Proposition}

It can be shown that the cumulant and partition functions are smooth $C^\infty$  analytic functions~\cite{brown1986fundamentals}.

A remarkable property is that strictly log-convex functions are also strictly convex:

\begin{Proposition}\label{prop:logcvxcvx}
A strictly log-convex function $Z:\Theta\subset\bbR^m\rightarrow\bbR$ is strictly convex. 
\end{Proposition}

\begin{proof}
By definition, function $Z(\theta)$ is strictly log-convex if and only if:
\begin{equation}\label{eq:proofzlogconvex}
\forall \theta_0\not=\theta_1,\quad Z(\alpha\theta_0+(1-\alpha\theta_1))< Z(\theta_0)^\alpha\, Z(\theta_1)^{1-\alpha},
\end{equation}
i.e., by taking the logarithm on both sides of the inequality, $F=\log Z$ is strictly convex:
\begin{eqnarray*}
\forall \theta_0\not=\theta_1,\quad \log Z(\alpha\theta_0+(1-\alpha)\theta_1) &<&\alpha\log Z(\theta_0)+(1-\alpha)\log Z(\theta_1),\\
\Leftrightarrow\quad F(\alpha\theta_0+(1-\alpha)\theta_1) &<& \alpha F(\theta_0) + (1-\alpha) F(\theta_1).
\end{eqnarray*}

Since $f(x)=\exp(x)$ is strictly convex (because $f''(x)=\exp(x)>0$), we have for all $\alpha\in (0,1)$:
$$
f(\alpha F(\theta_0)+(1-\alpha) F(\theta_1))< \alpha f(F(\theta_0)) + (1-\alpha) f(F(\theta_1)).
$$
Letting $F(\theta)=\log Z(\theta)$ in the above inequality, we get:
\begin{eqnarray}
\exp(\alpha \log Z(\theta_0)+(1-\alpha) \log Z(\theta_1) &<& \alpha\exp(\log Z(\theta_0)) + (1-\alpha) \exp(\log Z(\theta_1)),\\
Z(\theta_0)^\alpha\, Z(\theta_1)^{1-\alpha} &<&\alpha Z(\theta_0)+(1-\alpha) Z(\theta_1),\label{eq:zz}
\end{eqnarray}
and therefore we get from Eq.~\ref{eq:proofzlogconvex} and Eq.~\ref{eq:zz}:
\begin{equation}\label{eq:ineqZZ}
\forall \theta_0\not=\theta_1, Z(\alpha\theta_0+(1-\alpha\theta_1))<Z(\theta_0)^\alpha\, Z(\theta_1)^{1-\alpha}<\alpha Z(\theta_0)+(1-\alpha) Z(\theta_1).
\end{equation}
That is, $Z$ is strictly convex.

An alternative proof may   use the weighted arithmetic-mean geometric mean inequality~\cite{niculescu2006convex} (AM-GM inequality):
\begin{equation}\label{eq:amgm}
\forall a,b>0, \quad \alpha a+(1-\alpha)b\geq  a^\alpha b^{1-\alpha},
\end{equation}
with equality if and only if $a=b$.
Let $a=Z(\theta_0)$ and $b=Z(\theta_1)$ in the inequality of Eq.~\ref{eq:amgm}, we recover the inequality of Eq.~\ref{eq:zz}.
\end{proof}
See the Appendix for the weaker proposition of a log-convex function being convex using the assumption of second-order differentiability of convex functions.

\begin{figure}
\centering
\includegraphics[width=0.65\textwidth]{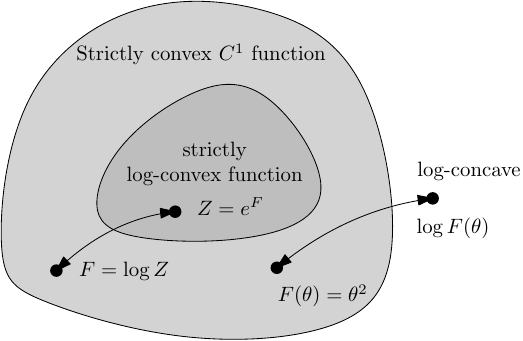}
\caption{Strictly log-convex functions form a proper subset of strictly convex functions.\label{fig:cvxlogcvx}}
\end{figure}

The converse of Proposition~\ref{prop:logcvxcvx} is not necessarily true: Some convex functions are not log-convex, and thus the class of strictly log-convex functions is a proper subclass of strictly convex functions.
For example, $\theta^2$ is convex but log-concave since $(\log \theta^2)''=-\frac{2}{\theta^2}<0$ (Figure~\ref{fig:cvxlogcvx}).

\begin{Remark}
Since $Z=\exp(F)$ is strictly convex (Proposition~\ref{prop:logcvxcvx}), $F$ is exponentially convex.
\end{Remark}

\begin{Definition}\label{prop:pair}
The  cumulant function $F$ and partition function $Z$  of a regular exponential family  are both strictly convex and smooth functions inducing a pair of dually flat spaces with corresponding Bregman divergences~\cite{Bregman-1967} $B_F$ (i.e., $B_{\log Z}$) and $B_Z$ (i.e., $B_{\exp F}$):
\begin{eqnarray}
B_Z(\theta_1:\theta_2)&=&Z(\theta_1)-Z(\theta_2)-\inner{\theta_1-\theta_2}{\nabla Z(\theta_2)}\geq 0,\\
B_{\log Z}(\theta_1:\theta_2)&=& \log \left(\frac{Z(\theta_1)}{Z(\theta_2)}\right) -\Inner{\theta_1-\theta_2}{{\frac{\nabla Z(\theta_2)}{Z(\theta_2)}}}\geq 0,
\end{eqnarray}
 and a pair of families of skewed Jensen divergences $J_{F,\alpha}$ and $J_{Z,\alpha}$:
\begin{eqnarray}
J_{Z,\alpha}(\theta_1:\theta_2) &=& \alpha Z(\theta_1)+ (1-\alpha)Z(\theta_2)-Z(\alpha\theta_1+(1-\alpha)\theta_2) \geq 0,\\
J_{\log Z,\alpha}(\theta_1:\theta_2)&=&  \log \frac{Z(\theta_1)^\alpha Z(\theta_2)^{1-\alpha}}{Z(\alpha\theta_1+(1-\alpha)\theta_2)} \geq 0.
\end{eqnarray}
\end{Definition}
 
For a strictly convex function $F(\theta)$, we define the symmetric Jensen divergence:
$$
J_F(\theta_1,\theta_2)=J_{F,\frac{1}{2}}(\theta_1:\theta_2)=\frac{F(\theta_1)+F(\theta_2)}{2}-F\left(\frac{\theta_1+\theta_2}{2}\right).
$$

Let $\calB_\Theta$ denote the set of real-valued strictly convex and differentiable functions defined on an open set $\Theta$, called Bregman generators.
We may equivalently consider the set of   strictly concave and differentiable functions $G(\theta)$, and let $F(\theta)=-G(\theta)$: 
See~\cite{wong2018logarithmic} (Equation~1).

\begin{Remark}
The non-negativeness of the Bregman divergences for the cumulant and partition functions define the criteria for checking the strict convexity 
or log-convexity of a $C^1$ function:
\begin{eqnarray*}
\mbox{$F(\theta)$ is strictly convex} &\Leftrightarrow& \forall \theta_1\not=\theta_2, B_F(\theta_1:\theta_2)>0,\\
&\Leftrightarrow&  \forall \theta_1\not=\theta_2, F(\theta_1)> F(\theta_2)+\inner{\theta_1-\theta_2}{\nabla F(\theta)},
\end{eqnarray*}
and
\begin{eqnarray*}
\mbox{$Z(\theta)$ is strictly log-convex} &\Leftrightarrow& \forall \theta_1\not=\theta_2, B_{\log Z}(\theta_1:\theta_2)>0,\\
&\Leftrightarrow& \forall \theta_1\not=\theta_2, \log Z(\theta_1) > \log Z(\theta_2) + \Inner{\theta_1-\theta_2}{ \frac{\nabla Z(\theta_2)}{Z(\theta_2)} }.
\end{eqnarray*}
\end{Remark}

The forward Bregman divergence $B_F(\theta_1:\theta_2)$ and reverse Bregman divergence $B_F(\theta_2:\theta_1)$ can be unified with the $\alpha$-skewed Jensen divergences by rescaling $J_{F,\alpha}$, and allowing $\alpha$ to range in $\bbR$~\cite{zhang2004divergence,nielsen2011burbea}:

\begin{equation}
J_{F,\alpha}^s(\theta_1:\theta_2)=
\left\{
\begin{array}{ll}
\frac{1}{\alpha(1-\alpha)} J_{F,\alpha}(\theta_1:\theta_2), & \alpha\in\bbR\backslash\{0,1\},\\
B_F(\theta_1:\theta_2), & \alpha =0,\\
4\, J_F(\theta_1,\theta_2), & \alpha=\frac{1}{2},\\
B_F^*(\theta_1:\theta_2)=B_F(\theta_2:\theta_1), & \alpha =1.
\end{array}
\right.,
\end{equation}
where $B_F^*$ denote the reverse Bregman divergence obtained by swapping the parameter order (reference duality~\cite{zhang2004divergence}).

\begin{Remark}
Alternatively, we may rescale $J_F$ by a factor  $\kappa(\alpha)=\frac{1}{\alpha(1-\alpha) 4^{4\alpha(1-\alpha)}}$, i.e., $J_{F,\alpha}^{\bar{s}}(\theta_1:\theta_2)=\kappa(\alpha)\, J_{F,\alpha}(\theta_1:\theta_2)$ so that
$\kappa(\frac{1}{2})=1$ and $J_{F,\frac{1}{2}}^{\bar{s}}(\theta_1:\theta_2)=J_F(\theta_1,\theta_2)$.
\end{Remark}

Next, we first recall the connections between these Jensen and Bregman divergences which are divergences between parameters and statistical divergence counterparts between probability densities in 
\S\ref{sec:divF}.
We will then introduce the novel connections between these parameter divergences and $\alpha$-divergences between unnormalized densities in 
\S\ref{sec:divZ}.

\section{Divergences related to the cumulant function}\label{sec:divF}

Consider the scaled $\alpha$-skewed Bhattacharyya distances~\cite{kailath1967divergence,nielsen2011burbea} between two probability densities $p(x)$ 
and $q(x)$:

$$
D_{B,\alpha}^s(p:q)=-\frac{1}{\alpha(1-\alpha)}\log\int p^\alpha q^{1-\alpha}\dmu,\quad \alpha\in\bbR\backslash\{0,1\}.
$$

The scaled $\alpha$-skewed Bhattacharyya distances can also be interpreted as R\'enyi divergences~\cite{van2014renyi} scaled by $\frac{1}{\alpha}$:
$D_{B,\alpha}^s(p:q)=\frac{1}{\alpha} D_{R,\alpha}(p:q)$ where

$$
D_{R,\alpha}(p:q)=\frac{1}{\alpha-1}\log \int p^\alpha q^{1-\alpha}\dmu.
$$

The Bhattacharyya distance $D_B(p,q)=-\log \int \sqrt{pq}\dmu$ corresponds to one fourth of $D_{B,\frac{1}{2}}^s(p:q)$: $D_B(p,q)=\frac{1}{4}\, D_{B,\frac{1}{2}}^s(p:q)$.
Since $D_{B,\alpha}^s$ tend to the Kullback-Leibler divergence $D_\KL$ when $\alpha\rightarrow 1$ and
to the reverse Kullback-Leibler divergence ${D_\KL}^*$ when $\alpha\rightarrow 0$, we have
$$
D_{B,\alpha}^s(p:q)=\left\{
\begin{array}{ll}
-\frac{1}{\alpha(1-\alpha)}\,\log\int p^\alpha q^{1-\alpha}\dmu, & \alpha\in\bbR\backslash\{0,1\},\\
D_\KL(p:q), & \alpha=1,\\
4\, D_B(p,q) &\alpha=\frac{1}{2},\\
{D_\KL}^*(p:q)=D_\KL(q:p) & \alpha=0.
\end{array}
\right.
$$

When both probability densities belong to the same exponential family $\calE=\{p_\theta(x) \st \theta\in\Theta\}$ with cumulant $F(\theta)$, we have the following proposition:

\begin{Proposition}[\cite{nielsen2011burbea}]
The scaled $\alpha$-skewed Bhattacharyya distances 
between two probability densities $p_{\theta_1}$ and $p_{\theta_2}$ of an exponential family amounts to the scaled $\alpha$-skewed Jensen divergence between their natural parameters:
\begin{equation}\label{eq:sbhat}
D_{B,\alpha}^s(p_{\theta_1}:p_{\theta_2})=J_{F,\alpha}^s(\theta_1,\theta_2).
\end{equation}
\end{Proposition}

\begin{proof}
The proof follows by first considering the $\alpha$-skewed Bhattacharyya similarity coefficient $\rho_\alpha(p,q)=\int p^\alpha q^{1-\alpha}\dmu$.

\begin{eqnarray*}
\rho_\alpha(p_{\theta_1}:p_{\theta_2}) &=& \int \exp\left(\inner{\theta_1}{x}-F(\theta_1)\right)^\alpha\, \exp\left(\inner{\theta_2}{x}-F(\theta_2)\right)^{1-\alpha}\dmu,\\
&=& \int \exp(\inner{\alpha\theta_1+(1-\alpha)\theta_2)}{x}) \, \exp\left(-(\alpha F(\theta_1)+(1-\alpha)F(\theta_2))\right)\dmu.
\end{eqnarray*}
Multiplying by $\exp(F(\alpha\theta_1+(1-\alpha)\theta_2))\, \exp(-F(\alpha\theta_1+(1-\alpha)\theta_2)) = \exp(0)=1$ the last equation, we get with $\bar\theta=\alpha\theta_1+(1-\alpha)\theta_2$:
$$
\rho_\alpha(p_{\theta_1}:p_{\theta_2})= \exp(-(\alpha F(\theta_1)+(1-\alpha)F(\theta_2)) \exp(F(\bar\theta)) \int \exp(\inner{\bar\theta}{x}-F(\bar\theta)) \dmu.
$$
Since $\bar\theta\in\Theta$, we have $\int \exp(\inner{\bar\theta}{x}-F(\bar\theta)) \dmu=1$, and therefore
$\rho_\alpha(p_{\theta_1}:p_{\theta_2})=\exp(-J_{F,\alpha}(\theta_1:\theta_2))$.
\end{proof}

For the practitioners in machine learning, it is well-known that the Kullback--Leibler divergence between two probability densities $p_{\theta_1}$ and $p_{\theta_2}$ of an exponential family  amounts to a Bregman divergence for the cumulant generator on swapped parameter order  (e.g., \cite{azoury2001relative} and~\cite{amari-1985}):
$$
D_\KL(p_{\theta_1}:p_{\theta_2}) = B_F(\theta_2:\theta_1).
$$

This is a particular instance of Eq.~\ref{eq:sbhat} obtained for $\alpha=1$:
$$
D_{B,1}^s(p_{\theta_1}:p_{\theta_2})=J_{F,1}^s(\theta_1,\theta_2).
$$

This formula has been further generalized in~\cite{nielsen2022statistical} by considering truncations of exponential family densities:
Let $\calX_1\subseteq\calX_2\subseteq\calX$ and $\calE_1=\{ 1_{\calX_1}(x)\, p_\theta(x)\}$, $\calE_2=\{1_{\calX_2}(x)\, q_{\theta'}(x)\}$ be two truncated families of $\calX$ with cumulant functions 
$$
F_1(\theta)=\log\int_{\calX_1} \exp(\inner{t(x)}{\theta})\dmu,
$$
 and 
$$
F_2(\theta')=\log\int_{\calX_2} \exp(\inner{t(x)}{\theta})\dmu\geq F_1(\theta').
$$
 Then we have
\begin{eqnarray*}
D_\KL(p_{\theta_1}:q_{\theta_2'}) &=& B_{F_2,F_1}(\theta_2':\theta_1),\\
&=& F_2(\theta_2')-F_1(\theta_1)-\inner{\theta_2'-\theta_1}{\nabla F_1(\theta_1)}.
\end{eqnarray*}

Truncated exponential families are normalized exponential families which may not be regular~\cite{del1994singly} (i.e., parameter space $\Theta$ may not be open).

\section{Divergences related to the partition function}\label{sec:divZ}

The squared Hellinger distance~\cite{IG-2016} between two positive potentially unnormalized densities $\tp$ and $\tq$ is defined by
\begin{eqnarray*}
D_H^2(\tp,\tq) &=& \frac{1}{2} \int (\sqrt{\tp}-\sqrt{\tq})^2\, \dmu,\\
&=& \frac{\int\tp\dmu+\int \tq\dmu}{2}-\int\sqrt{\tp\tq}\dmu
\end{eqnarray*}
Notice that the Hellinger divergence can be interpreted as the integral of the difference 
between the arithmetic mean $A(\tp,\tq)=\frac{\tp+\tq}{2}$ minus the geometric mean $G(\tp,\tq)=\sqrt{\tp\tq}$ of the densities:
 $D_H^2(\tp,\tq)=\int (A(\tp,\tq)-G(\tp,\tq)) \dmu$ . This also proves that $D_H(\tp,\tq)\geq 0$ since $A\geq G$.
The Hellinger distance $D_H$ satisfies the metric axioms of distances.

When considering unnormalized densities $\tp_{\theta_1}=\exp(\inner{t(x)}{\theta_1})$ and $\tp_{\theta_2}=\exp(\inner{t(x)}{\theta_2})$ of an exponential family $\calE$ with partition function $Z(\theta)=\int \tp_\theta\,\dmu$, we get
\begin{equation}
D_H^2(\tp_{\theta_1},\tp_{\theta_2})=\frac{Z(\theta_1)+Z(\theta_2)}{2}-Z\left(\frac{\theta_1+\theta_2}{2}\right)=J_Z(\theta_1,\theta_2), \label{eq:JZEF}
\end{equation}
since $\sqrt{\tp_{\theta_1}\tp_{\theta_2}}=\tp_{\frac{\theta_1+\theta_2}{2}}$.

The Kullback-Leibler divergence~\cite{IG-2016} extended to two positive densities $\tp$ and $\tq$ is defined by
\begin{equation}
D_\KL(\tp:\tq)=\int \left(\tp\log\frac{\tp}{\tq} +\tq-\tp\right)\dmu.
\end{equation} 
When considering unnormalized densities $\tp_{\theta_1}$ and $\tp_{\theta_2}$ of $\calE$, we get
\begin{eqnarray}
D_\KL(\tp_{\theta_1}:\tp_{\theta_2}) &=& \int \left(\tp_{\theta_1}(x)\log\frac{\tp_{\theta_1}(x)}{\tp_{\theta_2}(x)}+\tp_{\theta_2}(x)-\tp_{\theta_1}(x)\right)\dmu(x),\\
&=& \int \left( e^{\inner{t(x)}{\theta_1}} \inner{\theta_1-\theta_2}{t(x)} +  e^{\inner{t(x)}{\theta_2}}-e^{\inner{t(x)}{\theta_1}} \right)\, \dmu(x),\\
&=& \inner{\int t(x) e^{\inner{t(x)}{\theta_1}}\dmu(x) }{ \theta_1-\theta_2}+Z(\theta_2)-Z(\theta_1),\\
&=& \inner{\theta_1-\theta_2}{\nabla Z(\theta_1)}+Z(\theta_2)-Z(\theta_1)= B_{Z}(\theta_2:\theta_1), \label{eq:ekluef}
\end{eqnarray}
since $\nabla Z(\theta)=\int t(x)\tp_\theta(x)\dmu(x)$.
Let $D_\KL^*(\tp:\tq)=D_\KL(\tq:\tp)$ denote the reverse KLD.

More generally, the family of $\alpha$-divergences~\cite{IG-2016} between unnormalized densities $\tp$ and $\tq$ is defined for $\alpha\in\bbR$ by:
$$
D_\alpha(\tp:\tq)=
\left\{
\begin{array}{ll}
\frac{1}{\alpha(1-\alpha)} \int\left(\alpha\tp+(1-\alpha)\tq-\tp^\alpha\tq^{1-\alpha}\right)\dmu, & \alpha\not\in\{0,1\}\\
D_\KL^*(\tp:\tq)=D_\KL(\tq:\tp) & \alpha =0,\\
4\, D_H^2(\tp,\tq) & \alpha=\frac{1}{2},\\
D_\KL(\tp:\tq) & \alpha=1.
\end{array}
\right.
$$
We have $D_\alpha^*(\tp:\tq)=D_\alpha(\tq:\tp)=D_{1-\alpha}(\tp:\tq)$. and the $\alpha$-divergences are homogeneous divergences of degree $1$:
For all $\lambda>0$, $D_\alpha(\lambda\tq:\lambda\tp)=\lambda \, D_\alpha(\tq:\tp)$.
Moreoever, since $\alpha\tp+(1-\alpha)\tq-\tp^\alpha\tq^{1-\alpha}$ can be expressed as the difference of the weighted arithmetic minus the
 geometric means $A(\tp,\tq;\alpha;1-\alpha)-G(\tp,\tq;\alpha;1-\alpha)$, it follows from the arithmetic-geometric mean inequality that we have $D_\alpha(\tp:\tq)\geq 0$.

When considering unnormalized densities $\tp_{\theta_1}$ and $\tp_{\theta_2}$ of $\calE$, we get
$$
D_\alpha(\tp_{\theta_1}:\tp_{\theta_2})=
\left\{
\begin{array}{ll}
\frac{1}{\alpha(1-\alpha)} J_{Z,\alpha}(\theta_1:\theta_2), & \alpha\not\in\{0,1\}\\
B_Z(\theta_1:\theta_2)    & \alpha =0,\\
4\, J_Z(\theta_1,\theta_2) & \alpha=\frac{1}{2},\\
B_Z^*(\theta_1:\theta_2)=B_Z(\theta_2:\theta_1) & \alpha=1
\end{array}
\right.
$$

\begin{Proposition}\label{prop:Z}
The  $\alpha$-divergences between unnormalized densities of an exponential family amounts to scaled $\alpha$-Jensen divergences between their natural parameters for the partition function:
$$
D_\alpha(\tp_{\theta_1}:\tp_{\theta_2})= J_{Z,\alpha}^s(\theta_1:\theta_2).
$$
When $\alpha\in\{0,1\}$, the oriented Kullback-Leibler divergences between unnormalized exponential familiy densities amounts to reverse Bregman divergences on their corresponding natural parameters for the partition function:
$$
D_\KL(\tp:\tq)=B_Z(\theta_2:\theta_1).
$$ 
\end{Proposition}

\begin{proof}
For $\alpha\not\in\{0,1\}$, consider
$$
D_\alpha(\tp_{\theta_1}:\tp_{\theta_2}) = \frac{1}{\alpha(1-\alpha)} \int\left(\alpha\tp_{\theta_1}+(1-\alpha)\tp_{\theta_2}-\tp_{\theta_1}^\alpha\tp_{\theta_2}^{1-\alpha}\right)\dmu.
$$
We have $\int \alpha\tp_{\theta_1}\, \dmu=\alpha Z(\theta_1)$, $\int (1-\alpha)\tp_{\theta_2}\, \dmu=(1-\alpha)Z(\theta_2)$,
and $\int \tp_{\theta_1}^\alpha\tp_{\theta_2}^{1-\alpha}\dmu=\int \tp_{\alpha\theta_1+(1-\alpha)\theta_2}\dmu=Z(\alpha\theta_1+(1-\alpha)\theta_2)$.
It follows that $D_\alpha(\tp_{\theta_1}:\tp_{\theta_2})=\frac{1}{\alpha(1-\alpha)} J_{Z,\alpha}(\theta_1:\theta_2)== J_{Z,\alpha}^s(\theta_1:\theta_2)$.
\end{proof}

Notice that the KLD extended to unnormalized densities can be written as a generalized relative entropy, i.e., 
obtained as the difference of the extended cross-entropy minus the extended entropy (self cross-entropy):
\begin{eqnarray*}
D_\KL(\tp:\tq) &=& H^\times(\tp:\tq)-H(\tp),\\
&=& \int \left(\tp\log\frac{\tp}{\tq}+\tq-\tp\right)\dmu,
\end{eqnarray*}
with
$$
H^\times(\tp:\tq)=\int \left( \tp(x)\log\frac{1}{\tq(x)} + \tq(x)\right) \dmu(x) -1,
$$ 
and
$$
H(\tp)=H^\times(\tp:\tp)=\int \left( \tp(x)\log\frac{1}{\tp(x)} + \tp(x)\right) \dmu(x) -1. 
$$

\begin{Remark}
In general, let us consider two unnormalized positive densities $\tp(x)$ and $\tq(x)$, 
and let  $p(x)=\frac{\tp(x)}{Z_p}$ and $q(x)=\frac{\tq(x)}{Z_q}$ denote their corresponding normalized densities (with normalizing factors $Z_p=\int \tp\,\dmu$ and 
$Z_q=\int \tq\,\dmu$). 
The KLD between $\tp$ and $\tq$ can be expressed using the KLD between their normalized densities and normalizing factors:
\begin{equation}\label{eq:klekl}
D_\KL(\tp:\tq)=   Z_p\, \left(D_\KL(p:q) +\log \frac{Z_p}{Z_q} \right) + Z_q -Z_p.
\end{equation}
Similarly, we have
\begin{eqnarray}
H^\times(\tp:\tq) &=& Z_p\, H^\times(p:q)-Z_p\log Z_q+Z_q-1,\\
H(\tp) &=& Z_p\, H(p)-Z_p\log Z_p + Z_p-1,
\end{eqnarray}
and $D_\KL(\tp:\tq)=H^\times(\tp:\tq)-H(\tp)$.
\end{Remark}

Notice that Eq.~\ref{eq:klekl} allows us to derive the following identity between $B_Z$ and $B_F$:
\begin{eqnarray}
B_Z(\theta_2:\theta_1) &=& Z(\theta_1)\, B_F(\theta_2:\theta_1) +Z(\theta_1)\log\frac{Z(\theta_1)}{Z(\theta_2)}+Z(\theta_2)-Z(\theta_1),\\
&=& \exp(F(\theta_1))\, B_F(\theta_2:\theta_1) +(\exp F(\theta_1))(F(\theta_1)-F(\theta_2)) + \exp(F(\theta_2))-\exp(F(\theta_1)).
\end{eqnarray}
Let $D_\skl(a:b)=a\log\frac{a}{b}+b-a$ be the scalar KLD for $a>0$ and $b>0$. 
Then we rewrite Eq.~\ref{eq:klekl} as
$$
D_\KL(\tp:\tq)=   Z_p\, D_\KL(p:q) + D_\skl(Z_p:Z_q),
$$
and we have:
$$
B_Z(\theta_2:\theta_1)=Z(\theta_1)\, B_F(\theta_2:\theta_1)+D_\skl(Z(\theta_1):Z(\theta_2)).
$$
The KLD between unnormalized densities $\tp$ and $\tq$ with support $\calX$ can also be written as a definite integral of a scalar Bregman divergence:
$$
D_\KL(\tp:\tq)=\int_{\calX} D_\skl(\tp(x):\tq(x))\, \dmu(x)=\int_{\calX} B_{f_\skl}(\tp(x):\tq(x))\,\dmu(x),
$$
where $f_\skl(x)=x\log x-x$. Since $B_{f_\skl}(a:b)\geq 0 \forall a>0, b>0$, we deduce that $D_\KL(\tp:\tq)\geq 0$ with equality iff $\tp(x)=\tq(x)$ $\mu$-almost everywhere.

Notice that $B_Z(\theta_2:\theta_1) = Z(\theta_1)\, B_F(\theta_2:\theta_1) + D_\skl(Z(\theta_1):Z(\theta_2))$ can be interpreted as the sum of two divergences: a conformal Bregman divergence with a scalar Bregman divergence.

\begin{Remark}
Consider the KLD between the normalized $p_{\theta_1}$ and unnormalized $\tp_{\theta_2}$ densities of a same exponential family.
We have
\begin{eqnarray}
D_\KL(p_{\theta_1}:\tp_{\theta_2}) &=& B_F(\theta_2:\theta_1)-\log Z(\theta_2)+Z(\theta_2)-1,\\
&=& Z(\theta_2)-1-F(\theta_1)-\inner{\theta_2-\theta_1}{\nabla F(\theta_2)},\nonumber\\
&=& B_{Z-1,F}(\theta_2-\theta_1).
\end{eqnarray}
Divergence $B_{Z-1,F}$ is a duo Bregman pseudo-divergence~\cite{nielsen2022statistical}:
$$
B_{F_1,F_2}(\theta_1:\theta_2)=F_1(\theta_1)-F_2(\theta_2)-\inner{\theta_1-\theta_2}{\nabla F_2(\theta_2)},
$$
for $F_1$ and $F_2$ two strictly convex and smooth functions such that $F_1\geq F_2$.
Indeed, we check that both generators $F_1(\theta)=Z(\theta)-1$ and $F_2(\theta)=F(\theta)$ are Bregman generators, and we have $F_1(\theta)\geq F_2(\theta)$ since $e^x\geq x+1$ for all $x$ (with equality when $x=0$), i.e., $Z(\theta)-1\geq F(\theta)$.

More generally, the $\alpha$-divergences between $p_{\theta_1}$ and $\tp_{\theta_2}$ can be written as
\begin{equation}
D_\alpha(p_{\theta_1}:\tp_{\theta_2}) = \frac{1}{\alpha(1-\alpha)} \left(
\alpha Z(\theta_1)+(1-\alpha)-\frac{Z(\alpha\theta_1+(1-\alpha)\theta_2)}{Z(\theta_2)}.
\right),
\end{equation}
and the (signed) $\alpha$-skewed Bhattacharyya distances are given by
$$
D_{B,\alpha}(p_{\theta_1}:\tp_{\theta_2})=\log Z(\theta_2)-\log Z(\alpha\theta_1+(1-\alpha)\theta_2).
$$

\end{Remark}

Let us illustrate Proposition~\ref{prop:Z} with some examples:

\begin{Example}
Consider the family of exponential distributions $\calE=\{p_\lambda(x)=1_{x\geq 0}\, \lambda\exp(-\lambda x)\}$.
$\calE$ is an exponential family with natural parameter $\theta=\lambda$ and parameter space $\Theta=\bbR{>0}$, and sufficient statistic $t(x)=-x$.
The partition function is $Z(\theta)=\frac{1}{\theta}$ with $Z'(\theta)=-\frac{1}{\theta^2}$ and $Z''(\theta)=\frac{2}{\lambda^3}>0$.
The cumulant function is $F(\theta)=\log Z(\theta)=-\log\theta$.
with moment parameter $\eta=E_{p_\lambda}[t(x)]=F'(\theta)=-\frac{1}{\theta}$.
The $\alpha$-divergences between two unnormalized exponential distributions are:
\begin{equation}
D_\alpha(\tp_{\lambda_1}:\tp_{\lambda_2}) = 
\left\{
\begin{array}{ll}
\frac{1}{\alpha(1-\alpha)} J_{Z,\alpha}(\theta_1:\theta_2)=\frac{(\lambda_1-\lambda_2)^2)}{\alpha \lambda_1^2\lambda_2+(1-\alpha)\lambda_1\lambda_2^2} & \alpha\not\in\{0,1\}\\
D_\KL(\tp_{\lambda_2}:\tp_{\lambda_1})=B_Z(\theta_1:\theta_2) = \frac{(\lambda_1-\lambda_2)^2}{\lambda_1\lambda_2^2}   & \alpha =0,\\
4\, J_Z(\theta_1,\theta_2)= \frac{(\lambda_1-\lambda_2)^2}{2(\lambda_1\lambda_2^2+\lambda_1^2\lambda_2)} & \alpha=\frac{1}{2},\\
D_\KL(\tp_{\lambda_1}:\tp_{\lambda_2})=B_Z(\theta_2:\theta_1) = \frac{(\lambda_1-\lambda_2)^2}{\lambda_2\lambda_1^2}   & \alpha=1
\end{array}
\right.
\end{equation}
\end{Example}

\begin{Example}
Consider the family of univariate centered normal distributions with $\tp_{\sigma^2}(x)\propto \exp(-\frac{x^2}{2\sigma^2})$ and partition function $Z(\sigma^2)=\sqrt{2\pi\sigma^2}$ so that $p_{\sigma^2}(x)=\frac{1}{Z(\sigma^2)}\, \tp_{\sigma^2}(x) = \frac{1}{\sqrt{2\pi\sigma^2}} \exp(-\frac{x^2}{2\sigma^2})$.
We have natural parameter $\theta=\frac{1}{\sigma^2}\in\Theta=\bbR_{>0}$ and sufficient statistic $t(x)=-\frac{x^2}{2}$.
The partition function expressed with  the natural parameter is $Z(\theta)=\sqrt{\frac{2\pi}{\theta}}$ with $Z'(\theta)=-\sqrt{\frac{\pi}{2}} \theta^{-\frac{3}{2}}$ and 
$Z''(\theta)=\frac{3\sqrt{\pi}}{2^{\frac{3}{2}}} \theta^{-\frac{5}{2}}>0$ (strictly convex on $\Theta$).
The unnormalized KLD between $\tp_{\sigma^2_1}$ and $\tp_{\sigma^2_2}$ is
$$
D_\KL(\tp_{\sigma^2_1}:\tp_{\sigma^2_2})= B_Z(\theta_2:\theta_1) =\sqrt{\frac{\pi}{2}}\, \left(2\sigma_2-3\sigma_1+\frac{\sigma_1^3}{\sigma_2^2}\right).
$$
We check that we $D_\KL(\tp_{\sigma^2}:\tp_{\sigma^2})=0$.

For the Hellinger divergence, we have 
$$
D_H^2(\tp_{\sigma^2_1}:\tp_{\sigma^2_2})= J_Z(\theta_1,\theta_2) = \sqrt{\frac{\pi}{2}}(\sigma_1+\sigma_2)-2\sqrt{\pi}\frac{\sigma_1\sigma_2}{\sqrt{\sigma_1^2+\sigma_2^2}},
$$
and we check that $D_H(\tp_{\sigma^2}:\tp_{\sigma^2})=0$.

Consider the family of $d$-variate case of centered normal distributions with unnormalized density: 
$$
\tp_{\Sigma}(x)\propto \exp(-\frac{1}{2}x^\top \Sigma^{-1} x) = \exp\left(-\frac{1}{2}\tr(x^\top \Sigma^{-1} x) \right)
=\exp\left(-\frac{1}{2}\tr(x x^\top \Sigma^{-1} ) \right),
$$ 
obtained using the matrix trace cyclic property, and where $\Sigma$ is the covariance matrix.
We have $\theta=\Sigma^{-1}$ (precision matrix) and $\Theta=\Sym_{++}(d)$ for $t(x)=-\frac{1}{2}xx^\top$ with the matrix inner product $\inner{A}{B}=\tr(A^\top B)$.
The partition function $Z(\Sigma)=(2\pi)^{\frac{d}{2}} \sqrt{\det(\Sigma)}$ expressed with the natural parameter is
$Z(\theta)=(2\pi)^{\frac{d}{2}} \sqrt{\frac{1}{\det(\theta)}}$.
This is a convex function with 
$$
\nabla Z(\theta)=-\frac{1}{2} (2\pi)^{\frac{d}{2}} \frac{\nabla_\theta \det(\theta)}{\det(\theta)^{\frac{3}{2}}}=
-\frac{1}{2} (2\pi)^{\frac{d}{2}} \frac{\theta^{-1}}{\det(\theta)^{\frac{1}{2}}},
$$
since $\nabla_\theta \det(\theta)=\det(\theta)\theta^{-\top}$ using matrix calculus.

Consider now the family of univariate normal distributions 
$$
\calE=\left\{p_{\mu,\sigma^2}(x)=\frac{1}{\sqrt{2\pi\sigma^2}} \exp\left( -\frac{1}{2} \left(\frac{x-\mu}{\sigma}\right)^2\right) \right\}.
$$
Let $\theta=\left(\theta_1=\frac{1}{\sigma_2},\theta_2=\frac{\mu}{\sigma^2}\right)$ and 
$$
Z(\theta_1,\theta_2)=\sqrt{\frac{2\pi}{\theta_1}}\, \exp\left(\frac{1}{2}\frac{\theta_2^2}{\theta_1}\right).
$$
The unnormalized densities are $\tp_\theta(x)=\exp\left(-\frac{\theta_1x^2}{2}+x\theta_2\right)$.
We have 
$$
\nabla Z(\theta)=\left[\begin{array}{ll}
\sqrt{\frac{\pi}{2}} \frac{(\theta_1+\theta_2^2)\exp\left(\frac{\theta_2^2}{2\theta_1}\right)}{\theta_1^{\frac{5}{2}}}\\
\sqrt{2\pi}\frac{\theta_2\exp\left(\frac{\theta_2^2}{2\theta_1}\right)}{\theta_1^{\frac{3}{2}}}
\end{array}\right].
$$
It follows that $D_\KL[\tp_{\theta}:\tp_{\theta'}]=B_Z(\theta':\theta)$.
\end{Example}

\section{Deforming convex functions and their induced dually flat spaces}\label{sec:deforming}

\subsection{Comparative convexity}
The log-convexity can be interpreted as a special case of comparative convexity
with respect to a pair $(M,N)$ of comparable weighted means~\cite{niculescu2006convex}:

A function $Z$ is $(M,N)$-convex if and only if for $\alpha\in [0,1]$, we have 
\begin{equation}
Z(M(x,y;\alpha,1-\alpha)) \leq N(Z(x),Z(y);\alpha,1-\alpha),
\end{equation}
and strictly $(M,N)$-convex iff we have strict inequality for $\alpha\in(0,1)$ and $x\not=y$.
Furthermore, a function $Z$ is (strictly) $(M,N)$-concave if $-Z$ is (strictly) $(M,N)$-convex.

Log-convexity corresponds to $(A,G)$-convexity, i.e., convexity with respect to the weighted arithmetic and geometric means defined respectively by 
$A(x,y;\alpha,1-\alpha)=\alpha x+(1-\alpha)y$ and $G(x,y;\alpha,1-\alpha)=x^\alpha y^{1-\alpha}$.
Ordinary convexity is $(A,A)$-convexity.
 
A weighted quasi-arithmetic mean~\cite{Kolmogorov-1930} (also called Kolmogorov-Nagumo mean~\cite{komori2021unified}) is defined for a continuous and strictly increasing function $h$ by
$$
M_h(x,y;\alpha,1-\alpha)=h^{-1}(\alpha h(x)+(1-\alpha)h(x)).
$$
We let $M_h(x,y)=M_h\left(x,y;\frac{1}{2},\frac{1}{2}\right)$.
Quasi-arithmetic means include the arithmetic mean obtained for $h(u)=\id(u)=u$ and the geometric mean for $h(u)=\log(u)$, and
more generally, power means 
$$
M_p(x,y;\alpha,1-\alpha)=\left(\alpha x^p+(1-\alpha) y^p\right)^{\frac{1}{p}}=M_{h_p}(x,y;\alpha,1-\alpha),\quad p\not=0
$$
which are quasi-arithmetic means obtained for the family of generators $h_p(u)=\frac{u^p-1}{p}$ with inverse $h_p^{-1}(u)=(1+up)^{\frac{1}{p}}$.
In the limit $p\rightarrow 0$, we have $M_0(x,y)=G(x,y)$ for the generator $\lim_{p\rightarrow 0}h_p(u)=h_0(u)=\log u$.

\begin{Proposition}[\cite{aczel1947generalization,nielsen2017generalizing}]
A function $Z(\theta)$ is strictly $(M_\rho,M_\tau)$-convex with respect to two strictly increasing smooth functions $\rho$ and $\tau$ if and only if
 the function $F=\tau\circ Z\circ \rho^{-1}$ is strictly convex.
\end{Proposition}

Notice that the set of strictly increasing smooth functions form a non-abelian group with group operation the function composition, neutral element the identity function, and inverse element the functional inverse function.

Since log-convexity is $(A=M_\id,G=M_{\log})$-convexity, a function $Z$ is strictly log-convex iff $\log\circ Z\circ\id^{-1}=\log\circ Z$ is strictly convex.
We have
$$
Z=\tau^{-1}\circ F\circ \rho   \Leftrightarrow F=\tau\circ Z\circ \rho^{-1}.
$$

Starting from a given convex function $F(\theta)$, we can deform the function $F(\theta)$ to obtain a function $Z(\theta)$ using two strictly monotone functions $\tau$ and $\rho$: 
$Z(\theta)=\tau^{-1}(F(\rho(\theta)))$.

For a $(M_\rho,M_\tau)$-convex function $Z(\theta)$ which is also strictly convex, we can  define a pair of Bregman divergences $B_Z$ and $B_{F}$ with 
$F(\theta)=\tau(Z(\rho^{-1}(\theta)))$ and a corresponding pair of skewed Jensen divergences.

Thus we have the following generic deformation scheme:

{\footnotesize
$$
\underbrace{F=\tau\circ Z\circ \rho^{-1}}_{\mbox{$(M_{\rho^{-1}},M_{\tau^{-1}})$-convex when $Z$ is convex}} \xrightleftharpoons[\text{$(\rho^{-1},\tau^{-1})$-deformation}]{\text{$(\rho,\tau)$-deformation}} 
\underbrace{Z=\tau^{-1}\circ F\circ\rho}_{\mbox{$(M_\rho,M_\tau)$-convex when $F$ is convex}}
$$
}

In particular, when function $Z$ is deformed by strictly increasing power functions $h_{p_1}$ and $h_{p_2}$ for $p_1$ and $p_2$ in $\bbR$ as  
$$
Z_{p_1,p_2}=h_{p_2}\circ Z\circ h_{p_1}^{-1},
$$
$Z_{p_1,p_2}$  is strictly convex  when it is strictly $(M_{p_1},M_{p_2})$-convex, and thus induces corresponding Bregman and Jensen divergences.

\begin{Example}
Consider the partition function $Z(\theta)=\frac{1}{\theta}$ of the exponential distribution family ($\theta>0$ with $\Theta=\bbR_{>0}$).
Let $Z_p(\theta)=(h_p\circ Z)(\theta)=\frac{\theta^{-p}-1}{p}$. Then we have $Z_p''(\theta)=(1+p) \frac{1}{\theta^{2+p}}>0$ when $p> -1$.
Thus we can deform $Z$ smoothly by $Z_p$ while preserving convexity by ranging $p$ from $-1$ to $+\infty$.
We thus get a corresponding family of Bregman and Jensen divergences. 
\end{Example}

The proposed convex deformation by using quasi-arithmetic mean generators differs from the interpolation of convex functions using the technique of proximal average~\cite{bauschke2008proximal}.

Note that in~\cite{nielsen2017generalizing}, comparative convexity with respect to a pair of quasi-arithmetic means $(M_\rho,M_\tau)$ is used to define a
$(M_\rho,M_\tau)$-Bregman divergence which turns out to be equivalent to a conformal Bregman divergence on the $\rho$-embedding of the parameters. 

\subsection{Dually flat spaces}

We start by a refinement of the class of convex functions used  to generate dually flat spaces:

\begin{Definition}[Legendre type function~\cite{rockafellar1967conjugates}]\label{def:Legendre}
$(\Theta,F)$ is of Legendre type if the function $F:\Theta\rightarrow\bbR$ is strictly convex and differentiable with $\Theta\not=\emptyset$ and
\begin{equation}\label{eq:cond}
\lim_{\lambda\rightarrow 0} \frac{d}{\dlambda} F(\lambda\theta+(1-\lambda)\bar\theta)=-\infty,\quad \forall\theta\in\Theta, \forall\bar\theta\in\partial\Theta.
\end{equation}
\end{Definition}

Legendre-type functions $F(\Theta)$ admits a convex conjugate $F^*(\eta)$ via the
Legendre transform  $F^*(\eta)=\sup_{\theta\in\Theta} \inner{\theta}{\eta}-F(\theta)$:
$$
F^*(\eta)=\inner{\nabla F^{-1}(\eta)}{\eta}-F(\nabla F^{-1}(\eta)).
$$

A smooth and strictly convex function $(\Theta,F(\theta))$  of Legendre-type~\cite{rockafellar1967conjugates} induces a dually flat space~\cite{IG-2016} $\calM$, i.e., a smooth Hessian manifold~\cite{shima2007geometry} with a single global chart $(\Theta,\theta(\cdot))$~\cite{IG-2016}.
A canonical divergence $D(p:q)$  between two points $p$ and $q$ of $\calM$  is viewed as a single parameter contrast function~\cite{eguchi1985differential} $\calD(r_{pq})$ on the product manifold $\calM\times\calM$.
The canonical divergence and its dual canonical divergence $\calD^*(r_{qp})=\calD(r_{pq})$ can be expressed  equivalently either as dual Bregman divergences or as dual Fenchel-Young divergences (Figure~\ref{fig:canonicaldiv}):
\begin{eqnarray*}
\calD(r_{pq}) &=& B_F(\theta(p):\theta(q))=Y_{F,F^*}(\theta(p):\eta(q)),\\
&=& \calD^*(r_{qp})=B_{F^*}(\eta(q):\eta(p))=Y_{F^*,F}(\eta(q):\theta(p)),
\end{eqnarray*}
where $Y_{F,F^*}$ is the Fenchel-Young divergence:
$$
Y_{F,F^*}(\theta(p):\eta(q))=F(\theta(p))+F^*(\eta(q))-\inner{\theta(p)}{\eta(q)}.
$$

We have the dual global coordinate system $\eta=\nabla F(\theta)$ and domain $H=\{\nabla F(\theta)\st \theta\in\Theta\}$ which defines the dual 
Legendre-type potential function $(H,F^*(\eta))$. Legendre-type function ensures that ${F^*}^*=F$.

Manifold $\calM$ is called dually flat because the torsion-free affine connections $\nabla$ and $\nabla^*$ induced by the  potential functions $F(\theta)$ and $F^*(\eta)$ linked with the Legendre-Fenchel transformation are flat: That is,
their Christoffel symbols vanishes in the dual coordinate system: $\Gamma(\theta)=0$ and $\Gamma^*(\eta)=0$.

\begin{figure}
\centering
\includegraphics[width=0.9\textwidth]{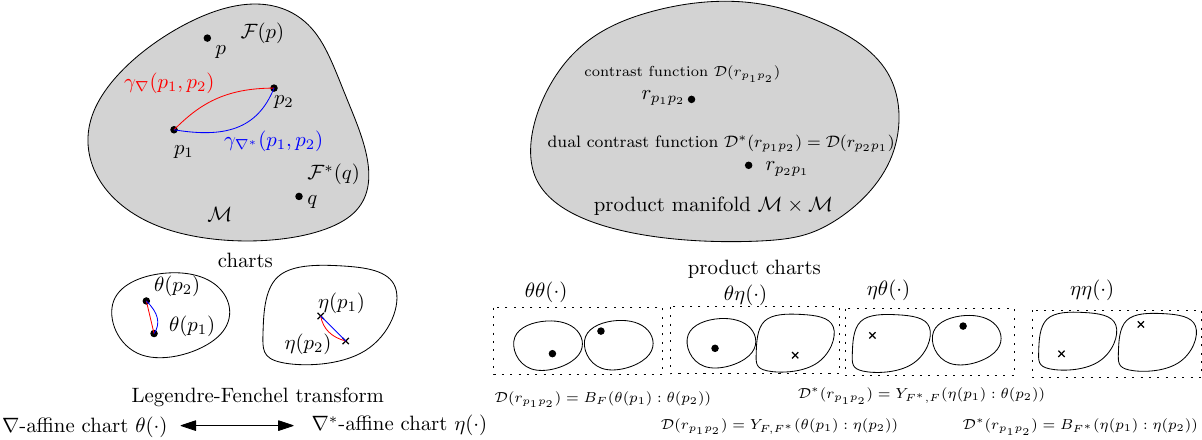}

\caption{The canonical divergence $\calD$ and the dual canonical divergence $\calD^*$ on a dually flat space $\calM$ equiped with potential functions $\calF$ and $\calF^*$ are viewed as single parameter  contrast functions on the product manifold $\calM\times\calM$:
The divergence $\calD$ can be expressed either using the $\theta\times\theta$-coordinate system as a Bregman divergence or using the mixed $\theta\times\eta$-coordinate system as a Fenchel-Young divergence.
Similarly, the dual divergence $\calD$ can be expressed either using the $\eta\times\eta$-coordinate system as a dual Bregman divergence or using the mixed $\eta\times\theta$-coordinate system as a dual Fenchel-Young divergence.\label{fig:canonicaldiv}
}
\end{figure}

The Legendre-type function $(\Theta,F(\theta))$ is not defined uniquely:
Function $\bar F(\bar\theta)=F(A\theta+b)+C\theta+d$ with $\bar\theta=A\theta+b$ for $A$ and $C$ invertible matrices and $b$ and $d$ vectors defines the same dually flat space with the same canonical divergence $D(p,q)$:
$$
D(p:q)=B_F(\theta(p):\theta(q))=B_{\bar F}(\bartheta(p):\bartheta(q)).
$$ 

Thus a log-convex Legendre-type function $Z(\theta)$ induces two dually flat spaces by considering the DFSs induced by $Z(\theta)$ and $F(\theta)=\log Z(\theta)$.
Let the gradient maps be $\eta=\nabla Z(\theta)$ and $\tilde\eta=\nabla F(\theta)=\frac{\eta}{Z(\theta)}$.

When $F(\theta)$ is chosen as the cumulant function of an exponential family, the Bregman divergence $B_F({\theta_1}:{\theta_2})$ can be interpreted as a statistical divergence between corresponding probability densities.
 Namely, the Bregman divergenec amounts to the reverse Kullback-Leibler divergence: $B_F({\theta_1}:{\theta_2})=D_\KL^*(p_{\theta_1}:p_{\theta_2})$ 
(see Appendix~\ref{sec:BFKLrev}).
 
Notice that deforming a convex function $F(\theta)$ into $F(\rho(\theta))$ such that $F\circ\rho$ remains strictly convex has been considered by  Yoshizawa and Tanabe~\cite{yoshizawa1999dual} to build a $2$-parameter deformation $\rho_{\alpha,\beta}$ of the dually flat space induced by the cumulant function $F(\theta)$ of the multivariate normal family.
See also the method of Hougaard~\cite{hougaard1983convex} to obtain other exponential families from a given exponential family.

Thus in general there are many more dually flat spaces with corresponding divergences and statistical divergences than the usually considered exponential family manifold~\cite{scarfone2014legendre} induced by the cumulant function.
It is interesting to consider their use in information sciences.

\section{Conclusion and discussion}\label{sec:concl}

For the machine learning practioner, it is well-known that the Kullback-Leibler divergence (KLD) between two probability densities $p_{\theta_1}$ and $p_{\theta_2}$ of an exponential family with cumulant function $F$ (free energy) amounts to a reverse Bregman divergence~\cite{azoury2001relative} induced by $F$ or equivalently to a reverse Fenchel-Young divergence~\cite{amari-1985}:
$$
D_\KL(p_{\theta_1}:p_{\theta_2})=B_F(\theta_2:\theta_1)=Y_{F,F^*}(\theta_2:\eta_1),
$$
where $\eta=\nabla F(\theta)$ is the dual moment or expectation parameter.

In this paper, we showed that the KLD extended to positive unnormalized densities  $\tp_{\theta_1}$ and $\tp_{\theta_2}$ of an exponential family with convex partition function $Z(\theta)$ (Laplace transform) amounts to a reverse Bregman divergence induced by $Z$ or equivalently to a reverse Fenchel-Young divergence:
$$
D_\KL(\tp_{\theta_1}:\tp_{\theta_2})=B_Z(\theta_2:\theta_1)=Y_{Z,Z^*}(\theta_2:\tilde\eta_1),
$$
where $\tilde\eta=\nabla Z(\theta)$.

More generally, we showed that the scaled $\alpha$-skewed Jensen divergences induced by the cumulant and partition functions between natural parameters coincide with 
 the scaled $\alpha$-skewed Bhattacharyya distances between probability densities and the $\alpha$-divergences between unnormalized densities, respectively:
\begin{eqnarray*}
D_{B,\alpha}^s(p_{\theta_1}:p_{\theta_2}) &=& J_{F,\alpha}^s(\theta_1:\theta_2) ,\\
D_\alpha(\tp_{\theta_1}:\tp_{\theta_2}) &=& J^s_{Z,\alpha}(\theta_1:\theta_2).
\end{eqnarray*}
We noticed that the partition functions $Z$ of exponential families are both convex and log-convex, and the corresponding cumulant functions are both convex and exponentially convex.

Figure~\ref{fig:summary} summarizes the relationships between statistical divergences between normalized and unnormalized densities of an exponential family and corresponding divergences between their natural parameters.
Notice that Brekelmans and Nielsen~\cite{Rob-2024} considered deformed uni-order likelihood ratio exponential families (LREFs) for annealing paths, and obtained an identity between the $\alpha$-divergences between unnormalized densities and Bregman divergences induced by multiplicatively scaled partition functions.  

\begin{figure}
\centering

\includegraphics[width=\textwidth]{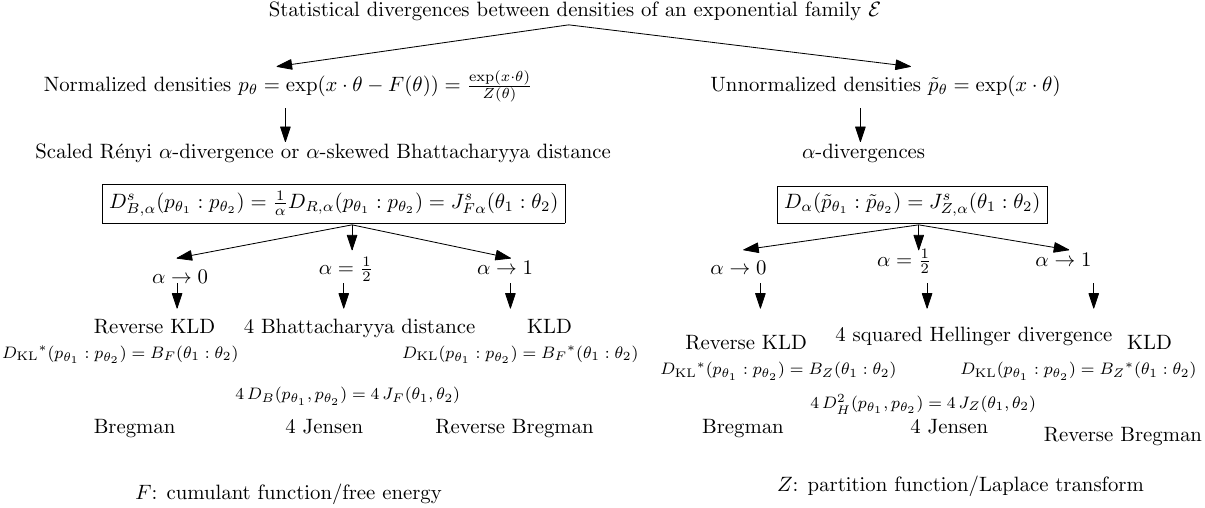}

\caption{Statistical divergences between normalized $p_\theta$ and unnormalized $\tp_{\theta}$ densities of an exponential family $\calE$ with corresponding divergences between their natural parameters.
Without loss of generality, we consider a natural exponential family (i.e., $t(x)=x$ and $k(x)=0$) with cumulant function $F$ and partition function $Z$.
$J_F$ and $B_F$ denote the Jensen and Bregman divergences induced by generator $F$, respectively.
Statistical divergences $D_{R,\alpha}$ and $D_{B,\alpha}$ denote the R\'enyi $\alpha$-divergences and the skewed $\alpha$-Bhattacharyya distances, respectively.
The superscript ``s'' indicates a rescaling by the multiplicative factor $\frac{1}{\alpha(1-\alpha)}$, and the superscript ``*'' denote the reverse divergence obtained by swapping the parameter order.
}\label{fig:summary}
\end{figure}

Since the log-convex partition function is also convex, we generalized the principle of building pairs of convex generators  using comparative convexity with respect to a pair of quasi-arithmetic means, and discussed about the induced dually flat spaces and divergences.
In particular, by considering the convexity-preserving deformations obtained by power mean generators, we show how to obtain a family of convex generators and dually flat spaces. 
Notice that some parametric families of Bregman divergences like the $\alpha$-divergences~\cite{amari2009alpha} or the $\beta$-divergences~\cite{hennequin2010beta} yield smooth families of dually flat spaces.

Banerjee et al.~\cite{banerjee2005clustering} proved a duality between regular exponential families and a subclass of Bregman divergences that they termed accordingly regular Bregman divergences. 
In particular, this duality allows one to view the Maximum Likelihood Estimator (MLE) of an exponential family with cumulant function $F$ as a right-sided Bregman centroid with respect to the Legendre-Fenchel dual  $F^*$.  
The scope of that duality was further extended for arbitrary Bregman divergences by introducing a class of generalized exponential families in~\cite{frongillo2014convex}.   

Concave deformations have also been recently studied in~\cite{ishige2022hierarchy}:
The authors introduce the $\log_\phi$-concavity induced by a positive continuous function $\phi$ generating a deformed logarithm $\log_\phi$ as the $(A,\log_\phi)$-comparative concavity (Definition 1.2~in\cite{ishige2022hierarchy}), and the weaker notion of $F$-concavity which corresponds to the $(A,F)$-concavity (Definition 2.1~in\cite{ishige2022hierarchy}, requiring strictly increasing functions $F$). 
Our deformation framework $Z=\tau^{-1}\circ F\circ \rho$ is more general since double-sided: 
We deform both the function $F$ by $F_\tau=\tau^{-1}\circ F$ and its argument $\theta$ by $\theta_\rho=\rho(\theta)$.

Exponentially concave functions are considered as generators of $L$-divergences in~\cite{wong2018logarithmic},
and $\alpha$-exponentially concave functions $G$ such that $\exp(\alpha G)$ are concave for $\alpha>0$ generalize the
$L$-divergences to $L_\alpha$-divergences which can be expressed equivalently using a generalization of the Fenchel-Young divergence based on the $c$-transforms~\cite{wong2018logarithmic}.
When $\alpha<0$, exponentially convex functions are considered instead of exponentially concave functions.
The information geometry induced by $L_\alpha$-divergences are dually projectively flat with constant curvature, and  
 reciprocally a dually projectively flat structure with constant curvature induces (locally) a canonical $L_{-\alpha}$-divergence.
Wong and Zhang~\cite{lambdaLegendre-2022} investigate a one-parameter deformation of convex duality called $\lambda$-duality by
 considering functions $f$ such that $\frac{1}{\lambda}(e^{\lambda f}-1)$ are convex for $\lambda\not=0$.
They define the $\lambda$-conjugate transform as a particular case of the $c$-transform~\cite{wong2018logarithmic}  and study the information geometry of the induced $\lambda$-logarithmic divergences. The $\lambda$-duality yields a generalization of exponential and mixture families to $\lambda$-exponential and $\lambda$-mixture families related to the R\'enyi divergence.

Finally, a class of statistical divergences called projective divergences are invariant under rescaling:
For example, the $\gamma$-divergence~\cite{fujisawa2008robust} $D_\gamma$ is such that $D_\gamma(p:q)=D_\gamma(\tp:\tq)$, and the $\gamma$-divergence tends to the KLD when $\gamma\rightarrow 0$.

\appendix

\section{Convexity of the cumulant functions of exponential families}\label{sec:appendix}
Let us prove Proposition~\ref{prop:naturalconvex} and Proposition~\ref{prop:FZconvexity}:

\noindent {\bf Proposition} \cite{barndorff2014information}
The natural parameter space $\Theta$ of an exponential family is convex.

\begin{proof}
Let $\Theta$ denote the natural parameter space:
$$
\Theta=\left\{\theta \st Z(\theta)=\int \exp(\inner{\theta}{x})\dmu<\infty\right\}
=\left\{\theta \st F(\theta)=\log\int \exp(\inner{\theta}{x})\dmu<\infty\right\}.
$$
Let $\theta_0, \theta_1\in\Theta$ and consider $\theta_\alpha=\theta_0+\alpha(\theta_1-\theta_0)$ for $\alpha\in (0,1)$.
In order to show that $\Theta$ is convex, we need to prove that $\theta_\alpha\in\Theta$, i.e., $Z(\theta_\alpha)<\infty$.
We have
\begin{eqnarray}
\int \exp(\inner{\theta_\alpha}{x}) \, \dmu(x) &=& \int \exp(\inner{\alpha\theta_0}{x})\,  \exp(\inner{(1-\alpha)\theta_1}{x})\dmu(x),\nonumber\\
&=&  \int \left(\exp(\inner{\theta_0}{x})\right)^{\alpha }\,  \left(\exp(\inner{\theta_1}{x})\right)^{(1-\alpha)}  \, \dmu(x). \label{eq:bholder}
\end{eqnarray}

Now, recall H\"older inequality for positive functions $f(x)$ and $g(x)$ with conjugate exponents $p$ and $q$ in $[1,\infty)$ such that 
$\frac{1}{p}+\frac{1}{q}=1$: 
$$
\int f(x) g(x)\dmu(x) \leq \left(\int f^p(x)\dmu(x)\right)^{\frac{1}{p}} \, \left(\int g^q(x)\dmu(x)\right)^{\frac{1}{q}}.
$$ 

Consider $f(x)=\left(\exp(\inner{\theta_0}{x})\right)^{\alpha}$ and $p=\frac{1}{\alpha}>1$ and $g(x)=\left(\exp(\inner{\theta_1}{x})\right)^{1-\alpha}$ with $q=\frac{1}{1-\alpha}>1$ (we check that $\frac{1}{p}+\frac{1}{q}=\alpha+1-\alpha=1$).
Thus we upper bound Eq.~\ref{eq:bholder} using H\"older inequality as follows:
\begin{equation}
\int \exp(\inner{\theta_\alpha}{x}) \, \dmu(x) \leq 
\left(\int \exp(\inner{\theta_0}{x}) \dmu(x)\right)^\alpha  \,
\left(\int \exp(\inner{\theta_1}{x}) \dmu(x) \right)^{1-\alpha} <\infty, \label{eq:ineqholder}
\end{equation}
since both  $\int \exp(\inner{\theta_0}{x}) \dmu(x)<\infty$ and $\int \exp(\inner{\theta_1}{x}) \dmu(x)<\infty$  because $\theta_0$ and $\theta_1$ both belong to $\Theta$. Hence, we have shown that $\Theta$ is convex.
\end{proof}

\noindent {\bf Proposition} \cite{barndorff2014information} 
The  cumulant function $F(\theta)$ is strictly convex and the partition function $Z(\theta)$  is positive and strictly log-convex.

\begin{proof}
 Eq.~\ref{eq:ineqholder} can be rewritten as
\begin{equation} \label{eq:zlogconvex}
Z(\theta_\alpha) \leq Z(\theta_0)^\alpha \, Z(\theta_1)^{1-\alpha}.
\end{equation}
Since the logarithm is a strictly increasing function, we have $a\leq b \Leftrightarrow \log(a)\leq \log (b)$ for any $a,b>0$, and it follows that:
\begin{eqnarray*}
\log Z(\theta_\alpha) &\leq& \log\left( Z(\theta_0)^\alpha \, Z(\theta_1)^{1-\alpha}\right),\\
F(\alpha\theta_0+(1-\alpha)\theta_1) &\leq & \alpha F(\theta_0) + (1-\alpha) F(\theta_1),
\end{eqnarray*}
hence the function $F(\theta)$ is strictly convex since inequality holds iff $\theta_0=\theta_1$.

Now, a function $Z(\theta)$ is strictly log-convex iff $\log Z(\theta)$ is strictly convex~\cite{niculescu2006convex}.
Since $Z(\theta)=\exp(F(\theta))$, we deduce that $Z(\theta)$ is positive and strictly log-convex since $\log Z(\theta)=F(\theta)$ is strictly convex.
\end{proof}

In Proposition~\ref{prop:logcvxcvx}, we prove that strictly log-convex functions are strictly convex.
Let us prove a weaker version of this proposition using the second-order differentiability of convex functions:

\noindent {\bf Proposition} 
A log-convex $C^2$ function $Z:\Theta\subset\bbR^m\rightarrow\bbR$ is convex. 
 
\begin{proof}
A $C^2$ function $f(x)$ is convex iff $f''(x)\geq 0$ but a $C^2$ function $f(x)$ with $f''(x)>0$ is strictly convex but not necessarily the converse: 
For example, $f(x)=x^4$ is strictly convex but $f''(x)=12x^2$ vanishes for $x=0$.
We can prove the convexity of log-convex functions by considering successively the univariate and multivariate cases as follows:

\begin{itemize}

\item In the univariate case $m=1$, consider the log-convex function $Z(\theta)=\exp({F(\theta)})$ for $F(\theta)$ a strictly convex function with $F''(\theta)\geq 0$.
Then we have $Z''(\theta)=F''(\theta)e^{F(\theta)}+(F'(\theta))^2 e^{F(\theta)}\geq 0$.

\item In the multivariate case $m>1$, $C^2$ function $F(\theta)$ is strictly convex if its Hessian is positive semi-definite.
Let $Z(\theta)$ be $C^2$ strictly log-convex, i.e., $F(\theta)=\log Z(\theta)$ is strictly convex.
Since $Z(\theta)=e^{F(\theta)}$ for $\nabla^2 F(\theta)\succeq 0$, we need to check that $\nabla^2 Z(\theta)\succeq 0$.
Let $\partial_l=\frac{\partial}{\partial\theta_l}$.
Then we have
\begin{eqnarray*}
\partial_i Z(\theta) &=& \partial_i e^{F(\theta)} = e^{F(\theta)}\, \partial_i F(\theta),\\
\partial_i \partial_j Z(\theta)& =& \partial_i \partial_j e^{F(\theta)} = e^{F(\theta)}\, \left( \partial_i\partial_j F(\theta)+(\partial_i F(\theta))(\partial_j F(\theta))\right),
\end{eqnarray*}
or in matrix form
$$
\nabla^2 Z(\theta)=e^{F(\theta)}\, \left(\nabla^2 F(\theta)+ \nabla F(\theta) (\nabla F(\theta))^\top )\right).
$$
Since (i) $e^{F(\theta)}>0$, (ii) $\nabla^2 F(\theta)\succeq 0$ and (iii) $\nabla F(\theta) (\nabla F(\theta))^\top\succeq 0$ 
(since a matrix outer-product is positive semi-definite: 
$\theta^\top \nabla F(\theta) (\nabla F(\theta))^\top \theta=(\theta^\top \nabla F(\theta))^2\geq 0$), we deduce that
$\nabla^2 Z(\theta)$ is positive semi-definite. Hence, $Z(\theta)$ is convex.
\end{itemize}
\end{proof}

\section{Statistical divergence recovered from the Bregman divergence induced by cumulant functions}\label{sec:BFKLrev}

\noindent {\bf Proposition}
The Bregman divergence $B_F({\theta_1}:{\theta_2})$ induced by the cumulant function $F(\theta)$ of an exponential family can be interpreted as a statistical divergence between corresponding probability densities $p_{\theta_1}$ and $p_{\theta_2}$:
 $B_F({\theta_1}:{\theta_2})=D_\KL^*(p_{\theta_1}:p_{\theta_2})$.

\begin{proof}
First, we prove that $F^*(\eta)=E_{p_{\theta}}[\log p_{\theta}]= -H(p_{\theta})$:
\begin{eqnarray*}
H(p_\theta) &=& -E_{p_\theta}[\log p_\theta],\\
&=& -E_{p_\theta}[\inner{\theta}{x}-F(\theta)],\\
&=& F(\theta)-\inner{\theta}{E_{p_\theta}[x]},\\
&=& F(\theta)-\inner{\theta}{\eta},\\
&=& -F^*(\eta).
\end{eqnarray*}
Notice that when $k(x)\not=0$, we have an additional term: $H(p_\theta)=-F^*(\eta)-E_{p_\theta}[k(x)]$.

Then using the linear property of expectation, we have
\begin{eqnarray*}
\inner{\theta_1}{\eta_2} &=& \inner{\theta_1}{E_{p_{\theta_2}}[x]},\\
&=& E_{p_{\theta_2}}[\inner{\theta_1}{x}]=E_{p_{\theta_2}}[\log \tp_{\theta_1}(x)],\\
&=& E_{p_{\theta_2}}[\log p_{\theta_1}(x)+F(\theta_1)],\\
&=& E_{p_{\theta_2}}[\log p_{\theta_1}(x)]+F(\theta_1).
\end{eqnarray*}

Thus it follows that the canonical divergence can be expressed in the mixed coordinate system as a Fenchel-Young divergence:

\begin{eqnarray*}
D(p_{\theta_1}:p_{\theta_2})= Y_{F,F^*}(\theta_1:\eta_2) &=& F(\theta_1) + F^*(\eta_2) -\inner{\theta_1}{\eta_2},\\
&=& F(\theta_1)-H(p_{\theta_2}) - E_{p_{\theta_2}}[\log p_{\theta_1}(x)] -F(\theta_1),\\
&=& H^\times(p_{\theta_2}:p_{\theta_1})-H(p_{\theta_2}),\\
&=& D_\KL(p_{\theta_2}:p_{\theta_1}),\\
&=& D_\KL^*(p_{\theta_1}:p_{\theta_2}).
\end{eqnarray*}
Thus the canonical divergence induced by the cumulant function amounts to a reverse Kullback-Leibler divergence $D_\KL^*$ between probability densities.
\end{proof} 

We can give statistical interpretations of the potential functions and their gradients as follows:
\begin{itemize}
\item $F(\theta)$ is the cumulant function (also called free energy in thermodynamics), 
\item $\eta=\nabla F(\theta)=E_{p_\theta}[t(x)]$ is the moment of the sufficient statistic,
\item $F^*(\eta)=-H(p_\theta)$ is the negentropy, and 
\item $\theta=\nabla F^*(\eta)$ are the Lagrangian multipliers in the maximum entropy problem~\cite{jaynes1957information}.
\end{itemize}


\end{document}